\documentclass[12pt,transaction, draftclsnofoot,onecolumn]{IEEEtran}
\usepackage{amsfonts}
\usepackage{amssymb}
\usepackage{graphicx}
\usepackage{subfigure}
\usepackage{enumerate}
\usepackage{amsmath}
\usepackage{color}
\usepackage{amsthm}
\usepackage{amsmath}
\usepackage{algorithm,,algpseudocode}
\usepackage[hyphens]{url}
\usepackage{hyperref}
\hypersetup{breaklinks=true}
\usepackage{stfloats}
\newcommand{\bp}{\begin{proof} \small }
\newcommand{\ep}{\end{proof} \normalsize}
\newcommand{\epx}{\end{proof} \small}
\newcommand{\bpa}{\begin{proofappx} \footnotesize }
\newcommand{\epa}{\end{proofappx} \small }
\newtheorem{theorem}{Theorem}
\newtheorem{proposition}{Proposition}

\newtheorem*{theorem*}{Theorem}
\newtheorem*{proposition*}{Proposition}
\newtheorem*{corollary*}{Corollary}
\newtheorem*{lemma*}{Lemma}
\newtheorem*{assumption*}{Assumption}
\newtheorem*{definition*}{Definition}
\newtheorem*{claim*}{Claim}

\newcommand{\be}{\begin{equation}}
\newcommand{\ee}{\end{equation}}
\newcommand{\bs}{\begin{subequations}}
\newcommand{\es}{\end{subequations}}
\newcommand{\bq}{\begin{eqnarray}}
\newcommand{\eq}{\end{eqnarray}}
\newcommand{\bqn}{\begin{eqnarray*}}
\newcommand{\eqn}{\end{eqnarray*}}

\newcommand{\ba}{\left[ \begin{array}}
\newcommand{\ea}{\\ \end{array} \right]}
\newcommand{\ben}{\begin{enumerate}}
\newcommand{\een}{\end{enumerate}}

\def\a{{\boldsymbol{a}}}
\def\b{{\boldsymbol{b}}}

\def\d{{\boldsymbol{d}}}

\def\q{{\boldsymbol{q}}}

\def\real{{\mathchoice%
{\hbox{\rm\setbox1=\hbox{I}\copy1\kern-.45\wd1 R}}
{\hbox{\rm\setbox1=\hbox{I}\copy1\kern-.45\wd1 R}}
{\hbox{\scriptsize\rm\setbox1=\hbox{I}\copy1\kern-.45\wd1 R}}
{\hbox{\scriptsize\rm\setbox1=\hbox{I}\copy1\kern-.45\wd1 R}}}}

\def\Zint{{\mathchoice{\setbox1=\hbox{\sf Z}\copy1\kern-.75\wd1\box1}
{\setbox1=\hbox{\sf Z}\copy1\kern-.75\wd1\box1}
{\setbox1=\hbox{\scriptsize\sf Z}\copy1\kern-.75\wd1\box1}
{\setbox1=\hbox{\scriptsize\sf Z}\copy1\kern-.75\wd1\box1}}}
\newcommand{\complex}{ \hbox{\rm C\kern-0.45em\rule[.07em]{.02em}{.58em}%
\kern 0.43em}}

\makeatletter
\newcommand{\algmargin}{\the\ALG@thistlm}
\makeatother
\newlength{\whilewidth}
\algdef{SE}[parWHILE]{parWhile}{EndparWhile}[1]
{\parbox[t]{\dimexpr\linewidth-\algmargin}{%
		\hangindent\whilewidth\strut\algorithmicwhile\ #1\ \algorithmicdo\strut}}{\algorithmicend\ \algorithmicwhile}%
\algnewcommand{\parState}[1]{\State%
	\parbox[t]{\dimexpr\linewidth-\algmargin}{\strut #1\strut}}

\ifodd 1

\else

\fi

\begin{document}
\title{Adaptive Fog Configuration for the Industrial Internet of Things}
%
%
%

\author{Lixing~Chen,
        Pan~Zhou,~\IEEEmembership{Member,~IEEE,}
        Liang~Gao,~\IEEEmembership{Member,~IEEE,}
        Jie~Xu,~\IEEEmembership{Member,~IEEE}
\thanks{Manuscript received January 10, 2018; revised April 03, 2018; accepted June 04, 2018; date of current version Jun 07, 2018. L. Chen and J. Xu are with the Department of Electrical and Computer Engineering, University of Miami. Email: lx.chen@miami.edu, jiexu@miami.edu. P. Zhou and L. Gao are with Huazhong University of Science and Technology (HUST). Email: panzhou@hust.edu.cn, gaoliang@mail.hust.edu.cn.}
}

\maketitle

\begin{abstract}
Industrial Fog computing deploys various industrial services, such as automatic monitoring/control and imminent failure detection, at the Fog Nodes (FNs) to  improve the performance of industrial systems. Much effort has been made in the literature on the design of fog network architecture and computation offloading. This paper studies an equally important but much less investigated problem of service hosting where FNs are adaptively configured to host services for Sensor Nodes (SNs), thereby enabling corresponding tasks to be executed by the FNs. The problem of service hosting emerges because of the limited computational and storage resources at FNs, which limit the number of different types of services that can be hosted by an FN at the same time. Considering the variability of service demand in both temporal and spatial dimensions, when, where, and which services to host have to be judiciously decided to maximize the utility of the Fog computing network. Our proposed Fog configuration strategies are tailored to battery-powered FNs. The limited battery capacity of FNs creates a long-term energy budget constraint that significantly complicates the Fog configuration problem as it introduces temporal coupling of decision making across the timeline. To address all these challenges, we propose an online distributed algorithm, called Adaptive Fog Configuration (AFC), based on Lyapunov optimization and parallel Gibbs sampling. AFC jointly optimizes service hosting and task admission decisions, requiring only currently available system information while guaranteeing close-to-optimal performance compared to an oracle algorithm with full future information.
\end{abstract}


%
\IEEEpeerreviewmaketitle

\section{Introduction}
The Industrial Internet of Things (IIoT) integrates the physical industrial environment into computer-based systems, resulting in improved efficiency, accuracy and economic benefit in addition to the reduced human intervention \cite{samad2016control}. Traditionally, industrial applications are founded on a centralized model of data processing and analytics: data generated by industrial devices are transported over the Internet infrastructure to a central computing facility (typically a cloud) where intensive data processing are carried out \cite{givehchi2013industrial}. However, the ever-growing distributed industrial data renders it impractical to transport all data over today's already-congested backbone Internet. Moreover, due to the unpredictable network latency, data processing in the cloud often cannot meet the stringent latency requirements of monitoring and controlling critical industrial devices \cite{suto2015energy}.

To overcome these limitations, Fog computing has recently been integrated into IIoT to support the operating environment featured by real-time response and high automation, which exploits the spare computing resources of edge devices (e.g. network gateways) to relieve backbone traffic burden and enable ultra-low latency response \cite{loshin2015intelligent}. In a typical scenario, industrial devices are equipped with smart Sensor Nodes (SNs) which collect data and perform first order operations (e.g. filtering, aggregation, and translation) on the raw data. A cohort of SNs can be logically clustered around and communicate with a Fog Node (FN) that provides a richer computing resource. These FNs receive streamed data from SNs, perform more complex analysis on the received data, and derive actionable intelligence to maintain the devices within its ``sphere of influence''. They also have the option of further offloading workload exceeding their computing capacity to the cloud as a second choice, resulting in a hierarchical Fog computing architecture \cite{azimi2017hich} (see Fig \ref{fig:sysillu} for an illustration). Task offloading has been a central theme of many prior works \cite{mao2017mobile,xu2017online}, which concerns what/when/how to offload workload of end devices to FNs. This literature assumes that FNs can process whatever types of service demand received without considering the availability of services. However, unlike the centralized cloud which has huge and diverse resources, the limited computing/storage resources of FNs allow only a small set of services to be hosted at the same time \cite{yang2016cost}. Because different industrial devices differ in their functionality and require different services to analyze the sensed data, which services are hosted by the FN determines which devices can be maintained at the network edge, thereby affecting the performance of Fog computing.

\begin{figure}[t]
	\centering
	\includegraphics[width= 0.55 \linewidth]{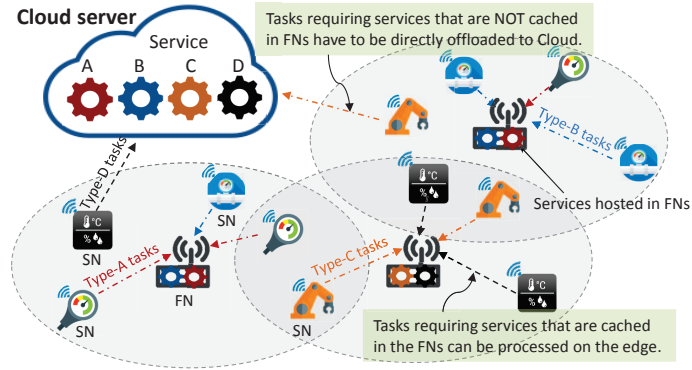}
	\caption{Illustration of adaptive Fog configuration}
	\label{fig:sysillu}
\end{figure}

Optimally configuring the Fog system (i.e., which services are hosted by which FNs) is a very challenging problem for IIoT. First, industrial devices are heterogeneous in terms of required service types and corresponding service demand. While the former is often fixed, the latter is changing over time since devices follow different operation schedules (regular maintenance or event-driven). Therefore, FNs must be adaptively configured to track the temporal variations of the demands for different services. Second, to accommodate more service demand at the Internet edge, FNs are usually densely deployed and hence an SN may be in the coverage of multiple FNs. On the one hand, the overlapping coverage allows FNs to collaboratively serve SN's demand. On the other hand, it creates a complex multi-cell setting where demand and resources are highly coupled in the spatial domain. Effective Fog configuration requires careful coordination among all FNs, and distributed solutions are in much favor. Third, FNs are deployed in a ``drop-and-play'' fashion to enable Fog computing on the existing infrastructure. In this scenario, FNs may not be powered by main electric grids but have to rely on batteries (or renewable energy sources) \cite{conti2017battery}. The battery energy constraints couple the fog configuration decisions over time, yet decisions have to be made without foreseeing the future system dynamics.
To address these challenges, we develop a novel online framework for adaptive Fog configuration.

In the conventional cloud computing context, virtual machine placement problems have been studied \cite{tordsson2012cloud}, and optimizing service placement over multiple cloudlets was investigated in \cite{yang2016cost}. However, these works do not consider the overlapping coverage areas or the battery energy constraints of FNs, and their algorithms are centralized. Content caching in Fog systems is also related to service hosting considered. While content caching \cite{shanmugam2013femtocaching} mainly deals with storage capacity constraints, our fog configuration strategy aims to improve computation delay performance withe energy budget constraints. Our main contributions are as follows:

(1) We formalize the adaptive Fog configuration problem for IIoT as a mixed-integer nonlinear stochastic optimization problem with long-term constraints. We jointly optimize service hosting and task admission of a network of FNs in order to minimize the time-average computation delay cost while satisfying the long-term battery energy constraints of FNs.

(2) A novel algorithm, called AFC, is developed for adaptive Fog configuration under the Lyapunov optimization framework \cite{neely2010stochastic}. AFC executes in an online fashion by separating the long-term problem into a sequence of per-slot subproblems that are solvable with only currently available system information. We prove that AFC achieves close-to-optimal service delay while bounding the potential violation of the long-term energy consumption constraint.

(3) To enable distributed coordination among a network of FNs, we develop a distributed algorithm as a key subroutine of AFC to solve each per-slot subproblem. The algorithm is developed based on Gibbs Sampling and leverages Markov Random Field and Graph Theory to enable parallel execution to speed up convergence.

The rest of the paper is organized as follows. Section \ref{sec:system model} presents the system model and problem formulation. Section \ref{sec:AFC} develops the online algorithm AFC and proves its performance guarantee. Section \ref{sec:CPGS} designs a distributed algorithm as a subroutine of AFC. Section\ref{sec:simulation} carries out simulations, followed by the conclusion in Section \ref{sec:conclusion}.
\section{System Model}\label{sec:system model}
\subsection{Industrial Fog System}
We consider an industrial environment comprising different types of devices. A hierarchical industrial Fog system is deployed to automate monitoring and control as well as apply embedded intelligent agents that can adjust device behaviors in relation to ongoing performance variables. Specifically, each device is equipped with a sensor node (SN), which is a low-power wireless device with embedded micro-controller and storage. Code for devices monitoring and control is deployed at SNs. Denote the set of SNs by $\mathcal{M} = \{1,...,M\}$. Besides the SNs, there are $N$ Fog Nodes (FNs) deployed in the network, indexed by $\mathcal{N} = \{1,...,N\}$, acting as wireless network gateways for connecting SNs and providing a richer computing resource that allows more complex analysis of streamed data from SNs for event triggering, predictive modeling of critical events, and notification. We consider that both SNs and FNs are battery-powered to enable flexible deployment.

Each FN serves the demand from SNs within its ``sphere of influence''. Let $\mathcal{M}_n\subset \mathcal{M}$ denote the set of SNs within the wireless transmission range of FN $n$. Due to dense deployment of FNs, an SN can be served by multiple FNs. Let $\mathcal{N}_m$ be the set of FNs reachable by SN $m$. We say that two FNs $i, j \in \mathcal{N}$ are neighbors if there exists some SN $m$ such that $i, j \in \mathcal{N}_m$. In other words, FNs $i$ and $j$ can potentially collaborate to serve at least one common SN. Given this, the Fog network can be described by a graph $G = \langle\mathcal{N}, \mathcal{E}\rangle$, where $\mathcal{E}$ is the edge set and there exists an edge between two FNs if they are neighbors. Let $\Omega_n \subseteq \mathcal{N}$ denote the one-hop neighborhood of FN $n$.

\subsection{Service Hosting and Task Admission}
Industrial SNs differ in functionalities and therefore, different SNs require different services to analyze different types of data. We consider that there are $K$ types of SNs in the network and hence $K$ types of services, indexed by $\mathcal{K} = \{1,2,...,K\}$. For each type-$k$ service demand, its (expected) input data size (in bits) and required number of CPU cycles for one task are $\gamma_k$  and $\mu_k$, respectively. Let $\theta(m) \in \mathcal{K}$ denote the type of SN $m$. Running a particular service requires allocating sufficient computing resource and caching the associated libraries and databases. However, compared to the powerful cloud, FNs are constrained in their computing resource and storage, hence only a limited number of services can be hosted by a FN at a time. We assume that each FN can host at most $C$ types of services. Hosting service $k$ at a FN allows \textit{in-situ} analysis of the streamed data from type-$k$ SNs, thereby enabling prompt response. The data from SNs, whose required services are not hosted at FNs, will be transmitted to the cloud for analysis.

Since industrial devices operate following different schedules, the service demand from SNs varies over time. To track such temporal variations, FNs adaptively reconfigure the hosted services to maximize the performance of the fog network. We consider a slotted operational timeline, where each time slot matches the time scale at which FNs can be reconfigured. The configuration decisions are made in a much slower time scale than task arrivals. During each time slot, task arrivals from SNs are assumed to follow Poisson processes, and the arrival rates in the current time slot are predicted using state-of-the-art prediction algorithms \cite{yoon2016requests}. Such two-scale time system is widely used in the existing literature \cite{maguluri2012stochastic}. The expected service demand of SN $m$ for type-$\theta(m)$ service in time slot $t$ is denoted by $d^t_m$. At the beginning of each time slot $t$, each FN configures itself by choosing what services to host. Let $\a^t_{n} = \{a^t_{n,k}\}_{k=1}^K$ be FN $n$'s (service) hosting decision in time slot $t$ where $a^t_{n,k}\in \{1,0\}$ is a binary variable representing whether service $k$ is hosted or not. The hosting decision has to satisfy the capacity constraint, namely $\sum_{k} a^t_{n,k} \leq C, \forall t$. Let $\mathcal{F}_i$ be the set of all feasible hosting decisions of FN $i$. The hosting profile of the whole Fog network is collected in $\a^t = \{\a^t_n\}_{n\in\mathcal{N}}$. Given the profile $\a^t$, let $\mathcal{B}^t_{m}(\a^t)\subseteq \mathcal{N}$ be the set of FNs that host service $\theta(m)$ and are reachable by SN $m$. We assume that the demand $d^t_m$ of SN $m$ is offloaded to FN $\mathcal{B}^t_{m}(\a^t)$ that has the best uplink channel condition, namely $\arg\max_{n \in \mathcal{B}^t_{m}(\a^t)} H^t_{m,n}$, where $H^t_{m,n}$ is the uplink channel condition between SN $m$ and FN $n$. In this way, SN $m$ incurs the least transmission energy consumption. Nevertheless, other SN-FN association rules can also be easily incorporated in our framework. It is possible that SN $m$ can reach none of the FNs hosting service $\theta(m)$, namely $\mathcal{B}^t_{m}(\a^t) = \emptyset$. In this case, the service demand of SN $m$ is sent to the reachable FN with the best channel condition and then further offloaded to the remote Cloud. To facilitate the exposition, we write $v^t_m(\a^t) \in \{ \mathcal{N} \cup 0\}$ as the FN (or Cloud) that processes the service demand for SN $m$ in time slot $t$:
\begin{equation}
\hspace{-0.1 in} v^t_m(\a^t) = \left\{
\begin{array}{ll}
\arg \max_{n \in \mathcal{B}^t_{m}(\a^t)} H_{m,n}, &\text{if}~\mathcal{B}^t_{m}(\a^t) \neq \emptyset\\
0, &\text{if}~\mathcal{B}^t_{m}(\a^t) = \emptyset
\end{array}\right.
\end{equation}

Because FNs are battery-powered, in addition to what services to host, they also decide the amount of workload to process by itself to extend the battery lifetime. Let $b^t_n \in [0,1]$ be the fraction of service demand admitted by FN $n$. Note that the actual task admission will be decided during the time slot when the specific service demand is received depending on its type and priority. Nevertheless, the task admission decisions can still be planned at a reasonably high granularity at the beginning of each time slot. We collect the task admission decisions of all FNs in time slot $t$ in $\b^t = \{b^t_n\}_{n \in \mathcal{N}}$.

\subsection{Energy Consumption and Service Delay}
Different service hosting and task admission decisions shape the service demand distribution among the FNs and the Cloud in different ways, resulting in different energy consumption of the FNs and service delay. Let $\lambda^t_{n,k}$ be the type-$k$ demand received by FN $n$, which can be computed as:
\begin{align}
\lambda^t_{n,k}(\a^t) = \sum_{m: \theta(m) = k, v^t_m(\a^t) = n} d^t_m(\a^t) .
\end{align}

\subsubsection{Energy consumption}
To simplify our analysis, we assume that the FN processes tasks at its maximum CPU speed and chooses the minimum CPU speed when it is idle. Then, based on the energy consumption model in \cite{mao2017mobile} the computation energy consumption can be expressed as:
\begin{align}
E^t_n(\a^t,\b^t) = E^o_n + \kappa_n b^t_n \sum\nolimits_{k}\mu_k \lambda^t_{n,k}(\a^t),
\end{align}
where $E^o_n$ is the static energy consumption regardless of the workload as long as FN $n$ is turned on; $\kappa_n = cf_n^2$ is the unit energy consumption for one CPU cycle depending on the CPU architecture parameter $c$ and CPU frequency $f_n$; $b^t_n \sum_{k}\mu_k \lambda^t_{n,k}(\a^t)$ is the total number of CPU cycles required to process service demand received by FN $n$.

\subsubsection{Service delay}
The service delay consists of computation delay and communication delay. The computation delay is incurred by task processing that happens at either the FNs or cloud. Following the Fog server computation model in \cite{lyu2017multiuser}, the computation delay for one type-$k$ service demand is $\mu_k/f_n$ if it is processed at FN $n$ and is $\mu_k/f_0$ if it is processed at Cloud, where $f_n$ and $f_0$ is the CPU frequency at FN $n$ and Cloud, respectively. Usually, we have $f_0 > f_n$.

The communication delay are incurred during wireless transmission between SNs and FNs and, wired transmission between FNs and cloud via the backbone Internet. Since the wireless transmission delay for each SN-FN pair is similar and much smaller compared to backbone transmission \cite{hu2016quantifying}, their impact on the service delay can be neglected. Therefore, we focus on the backbone transmission delay. Let $r^t$ be the backbone transmission rate and $h^t$ be the Round-Trip Time to Cloud, the service delay cost for SN $m$ can be obtained as:
\begin{align}
&D^t_m(\a^t,\b^t) =\\
&\left\{
\begin{array}{ll}
b^t_n d^t_m \dfrac{\mu_{\theta(m)}}{f_n}\\ + (1-b^t_n) d^t_m \left(\dfrac{\mu_{\theta(m)}}{f_{0}} + \dfrac{\gamma_{\theta(m)}}{r^t} + h^t\right), &\text{if}~~v^t_m(\a^t) = n\\
d^t_m \left(\dfrac{\mu_{\theta(m)}}{f_{0}} + \dfrac{\gamma_{\theta(m)}}{r^t} + h^t\right),&\text{if}~~v^t_m(\a^t) = 0
\end{array}\right.\nonumber
\end{align}
where, in the first case, the first term captures the service delay for demands processed at FN $v^t_m(\a^t)=n$ and the second term captures the service delay for demands processed at Cloud.

\subsection{Offline Problem Formulation}
The goal of Fog network is to minimize the total service delay cost of all SNs while satisfying the energy consumption constraints of FNs. Formally, the offline problem is
\begin{subequations}
	\begin{align}
	\textbf{P1}~~~&\min_{\a^t,\b^t, \forall t } ~~\dfrac{1}{T}\sum_{t=0}^{T-1}\sum_{m=1}^M D^t_m(\a^t, \b^t)\\
	\text{s.t.}~~&\dfrac{1}{T}\sum_{t=0}^{T-1} E^t_n(\a^t, \b^t) \leq Q_n, \forall n \label{cons:LT_eng}\\
	&\sum_{k} a^t_{n,k} \leq C, \forall n, \forall t \label{cons:capacity}\\
	& E^t_n(\a^t, \b^t) \leq E^{\max}_n, \forall n, \forall t \label{cons:E_max}\\
	& \sum_{m=1}^M D^t_m(\a^t, \b^t) \leq D^{\max}, \forall t \label{cons:D_max}
	\end{align}
\end{subequations}
where the first constraint is the long-term energy consumption constraint for each FN, and $T \cdot Q_n$ is the available battery energy of FN $n$ for a period of $T$ time slots. Every $T$ time slots the battery is replenished either manually or via an energy harvesting device. The second constraint is due to FNs' service hosting capacity. The third and fourth conditions impose per-slot constraints on the maximum energy consumption of each FN and the maximum total service delay cost.

There are several challenges that impede the derivation of the optimal solution to the offline problem \textbf{P1}. First, optimally solving \textbf{P1} requires the complete future information (e.g., service demands for all $t$) which is difficult to predict in advance, if not impossible. Second, the long-term energy constraint couples the configuration decisions temporally: consuming more energy in the current slots will reduce the available energy for future use. Third, \textbf{P1} is a mixed integer nonlinear programming which is very difficult to solve even if the future information is known a priori. These challenges call for an efficient online approach that can make Fog configuration decisions with only currently available information.

\section{Online Adaptive Fog Configuration} \label{sec:AFC}
In this section, we develop an online algorithm AFC (Adaptive Fog Configuration) based on Lyapunov optimization. A salient advantage of AFC is that it converts the offline problem \textbf{P1} to a sequence of per-slot optimization problems that are solvable with only currently available information.

\subsection{Online Fog Configuration with Lyapunov Drift}
A major challenge of directly solving \textbf{P1} is that the long-term energy constraint of FNs couples the Fog configuration decisions and task admission decisions across different time slots. To address this challenge, we leverage the \emph{Lyapunov drift} technique and construct a set of (virtual) energy deficit queues $\q(t)=\{q_n(t)\}_{n\in\mathcal{N}}$, one for each FN, to guide the Fog configuration and task admission to follow the long-term energy constraints \eqref{cons:LT_eng}. Initializing $q_n(0) = 0, \forall n\in\mathcal{N}$, the energy deficit queue for FN $n$ evolves as follows
\begin{align}\label{eq:queue_update}
q_n(t+1) = \max\left\{q_n(t) + E^t_n(\a^t, \b^t) - Q_n,0\right\},
\end{align}
The length of $q_n(t)$ indicates the deviation of the current energy consumption of FN $n$ from its long-term energy constraint $Q_n$. Based on the energy deficit queues, we present the online algorithm AFC in Algorithm \ref{alg:AFC}. AFC determines the optimal service hosting profile $\a^{*,t}$ and task admission profile $\b^{*,t}$ in each time slot $t$ by solving the following problem:
\begin{subequations}\label{P2}
	\begin{align}
	\textbf{P2}~~\min_{\a^t, \b^t}~ V \sum_{m=1}^{M} & D^t_m(\a^t, \b^t) + \sum_{n=1}^N q_n(t) E^t_n(\a^t,\b^t) \label{eq:p2_obj}\\
	\text{s.t.}~~&\eqref{cons:capacity}, \eqref{cons:E_max}, \eqref{cons:D_max}
	\end{align}
\end{subequations}
where $V$ is a positive control parameter used to adjust the trade-off between service delay minimization and energy deficit minimization. Note that solving \textbf{P2} requires only currently available information. By considering the additional term $\sum_n q_n(t) E^t_n(\a^t, \b^t)$, AFC takes into account the energy deficit of FNs in decision making. When $\q(t)$ is larger, minimizing the energy consumption is more critical. Thus, AFC works by following the philosophy of ``if violate the energy constraint, then consume less energy'', thereby satisfying the long-term energy constraints without foreseeing the future.
\begin{algorithm}[tb]
	\caption{AFC: Adaptive Fog Configuration}
	\begin{algorithmic}[1]
		\Statex \textbf{Input}: $q_n(0) \gets 0, \forall n \in \mathcal{N}$, $\mu_k$, $\gamma_k$, $C$, $E_n^{\max}$,$D^{\max}$;
		\For{$t=0~\text{to}~T-1$}
		\State Predict service demand $d^t_{m}, \forall m\in\mathcal{M}$;
		\parState{Observe backhaul transmission rate $r^t$, round-trip time $h^t$, uplink channel condition $H^t_{m,n}$;}
		\State Obtain $\a^{*,t}$, $\b^{*,t}$ by solving \textbf{P2}.
		\State Update energy deficit queue for each FN: \Statex \qquad $q_n(t+1) = \max\left\{q_n(t) + E^t_n(\a^t, \b^t) - Q_n,0\right\}$;
		\EndFor
		\Statex \textbf{Return}: $\{\a^{*,t}\}_{t=0}^{T-1}$ and $\{\b^{*,t}\}_{t=0}^{T-1}$.
	\end{algorithmic}\label{alg:AFC}
\end{algorithm}

\subsection{Performance Analysis of AFC}\label{sec:performance}
Next, we provide the performance bound of AFC in terms of the long-term service delay cost and long-term energy consumption compared to the optimal solution of \textbf{P1} obtained by an oracle with full future information.
\begin{theorem}\label{theo:AFC_bound}
	By following the Fog configuration profile $\{\a^{*,t}\}_{t=0}^{T-1}$ and task admission decisions $\{\b^{*,t}\}_{t=0}^{T-1}$ derived by AFC, the long-term service delay cost satisfies:
	\begin{equation*}\label{system_delay_bound}
	\lim_{T\rightarrow\infty}\dfrac{1}{T}\sum_{t=0}^{T-1}\sum_{m=1}^{M}\mathbb{E}\left[D^t_m(\a^{*,t},\b^{*,t})\right] < D^{\text{opt}}+\dfrac{B}{V},
	\end{equation*}
	and the long-term energy deficit of FNs satisfies:
	\begin{align*}\label{energy_defit_bound}
	\lim_{T\rightarrow\infty}\dfrac{1}{T}\sum_{t=0}^{T-1}\sum_{n=1}^{N}&\left(\mathbb{E}\left[E_n^t(\a^{*,t},\b^{*,t})\right]-Q_n\right)\\&\qquad\leq\dfrac{B}{\epsilon}+\dfrac{V}{\epsilon}(D^{\max}-D^{\text{opt}}),
	\end{align*}
	where $B=\frac{1}{2}\sum_{n=1}^{N}\left(E_n^{\max}-Q_n\right)^2$, $D^{\text{opt}}$ is the long-term service delay cost achieved by the optimal solution to \textbf{P1}; and $\epsilon>0$ is a constant which represents the long-term energy surplus achieved by some stationary policy.
\end{theorem}
\begin{proof}
	See online Appendix A \cite{onlineappendix}.
\end{proof}	

Theorem \ref{theo:AFC_bound} demonstrates an $[O(1/V ), O(V )]$ \emph{delay-energy deficit}  trade-off. Specifically, the asymptotic expected system delay cost achieved by AFC is no higher than the optimal delay performance of the offline problem \textbf{P1} plus a term $B/V$ where $B$ is a constant. Therefore, by choosing a large $V$, AFC is able to achieve the optimal system delay cost. However, a lower service delay is achieved at the price of a higher energy consumption. As presented in \eqref{energy_defit_bound}, the expected energy deficit is bounded by $[B+V(D^{\text{max}}-D^{\text{opt}})]/\epsilon$ and hence a large $V$ may incur a large energy consumption. To complete the algorithm, it remains to solve the optimization problem \textbf{P2}, which will be discussed in the next section.

\section{Distributed Optimization with Gibbs Sampler} \label{sec:CPGS}
The problem \textbf{P2} is a mixed-integer nonlinear programming. While there exist various techniques (such as Generalized Benders Decomposition) to solve it, these methods are usually centralized. When a centralized control is absent for information collection and centralized coordination, a distributed solution is desired so that each FN or (a small subset of FNs) can be configured in a distributed way. In the following, we develop a distributed algorithm based on Gibbs sampling techniques \cite{geman1984stochastic} for solving \textbf{P2}. Since \textbf{P2} is solved in each time slot $t$, we drop the time index in the rest of this section.

\subsection{Gibbs Sampling for Fog Configuration}
 In problem \textbf{P2}, the service hosting profile $\a$ and task admission profile $\b$ have to be jointly optimized. Fortunately, the task admission decisions of FNs are fully decoupled if the Fog configuration profile $\a$ is determined in advance. The optimal $\b$ can be easily derived because \textbf{P2} is a linear programming with a fixed $\a$. Since each $\a$ associates with an optimal $\b$, we denote $\b$ as a deterministic function of $\a$, denoted by $\tilde{\b}(\a) = \arg \min_{\b} V \sum_{m=1}^{M}  D_m(\a, \b) + \sum_{n=1}^N q_nE_n(\a,\b)$. For ease of exposition, we write $\tilde{\b}(\a)$ as $\tilde{\b}$ in the rest of the paper. Now, we restate \textbf{P2} with only $\a$ as a variable:
\begin{subequations}
	\begin{align}
	\textbf{P2-S}~\min_{\a} \mathcal{L}(\a)  =  V  & \sum_{m=1}^{M} D_m(\a, \tilde{\b}) + \sum_{n=1}^N q_n E_n(\a,\tilde{\b}) \\
    \text{s.t.}~~~&\eqref{cons:capacity}, \eqref{cons:E_max}, \eqref{cons:D_max}
	\end{align}
\end{subequations}

\textbf{P2-S} is a non-convex combinatorial optimization and it cannot be solved with many distributed algorithms based on Alternating Direction Method of Multipliers (ADMM) and Dual Decomposition \cite{boyd2011distributed}, since these methods usually require the problem to be convex. While there exist distributed algorithms for non-convex combinatorial optimization, e.g. Distributed Stochastic Algorithm \cite{zhang2002distributed}, most of them provide only convergence to local minimum. In the following, we leverage Gibbs Sampling (GS) to solve \textbf{P2-S}. The key advantage of GS is that it is able to converge to the global optimum with cooling schedule \cite{welsh1989simulated}. The GS is carried out in an iterative manner: in each iteration, an FN is selected (randomly or according to a predetermined order) to sample a new hosting decision from the conditional distribution with the remaining FNs fixing their decisions. The theory of Markov chain Monte Carlo guarantees that the probability of choosing a service hosting profile $\a$ is proportional to $e^{-\mathcal{L}(\a)/\sigma}$, which is known as the Gibbs distribution. Moreover, running GS while reducing the temperature parameter $\sigma$ can obtain the optimal $\a$ that minimizes the objective \cite{welsh1989simulated}. However, conventional GS is performed in sequential (i.e., one FN updates decision at a time) and the sequential GS has two main drawbacks: (1) it takes too long to complete one round of decision updating for large networks, (2) it works with an additional assumption that the global communication for information exchange is available, which may not hold in distributed Fog systems. To address these problems, we propose Chromatic Parallel Gibbs Sampling (CPGS) to enable the distributed decision making. It is worth noting that the convergence results of GS are typically available for sequential Gibbs samplers and extreme parallelism often cannot guarantee the ergodicity and convergence to Gibbs distribution \cite{newman2008distributed}. The proposed CPGS carefully transforms sequential GS into an equivalent parallel sampling by exploiting the special structure of considered Fog network using \emph{Markov Random Field} and \emph{Graph Coloring}, such that the ergodicity and convergence are preserved.

\subsection{Chromatic Parallel Gibbs Sampling}
To reach the Gibbs distribution, the sequential GS samples hosting decision $a_i(\tau+1)$ at FN $i$ for iteration $(\tau+1)$ according to a posterior conditional distribution $p$ calculated based on the hosting profile $\a(\tau)$ in iteration $\tau$ as follows:
\begin{align*}
a_i(\tau+1) \sim p(a_i | \a_{-i}(\tau)) \triangleq \frac{\exp(-\mathcal{L}(\{a_i, \a_{-i}(\tau)\})/\sigma)}{\sum\limits_{a_i^\prime\in\mathcal{F}_i}\exp(-\mathcal{L}(\{a_i^\prime, \a_{-i}(\tau)\})/\sigma)}
\end{align*}
where $\a_{-i}$ refers to the hosting decisions of FNs excluding the FN $i$. Intuitively, if the decision update at FN $i$ does not affect update of FN $j$ and vice versa, (i.e., $p(a_i|\a_{-i})$ and $p(a_j|\a_{-j})$ are independent), then FNs $i$ and $j$ can update their decisions independently and simultaneously. To formalize this, we resort to the Markov Random Field (MRF). An MRF is an undirected graph over FNs' hosting decisions. On this graph, the set of decisions adjacent to $a_i$, denoted by $\Gamma_i$, is called the Markov Blanket of $a_i$. Given $\Gamma_i$, the decision update of FN $i$ is conditionally independent of FNs outside $\Gamma_i$, namely $p(a_i | \a_{-i})=p(a_i | \a_{\Gamma_i})$. Therefore, any two FNs that are not in the Markov Blanket of each other can evolve their decisions simultaneously. The next proposition establishes a connection between the physical fog network and the MRF.

\begin{proposition}\label{theo:MFR}
	For FN $i\in\mathcal{N}$, the Markov Blanket $\Gamma_i$ of $a_i$ consists of the service hosting decisions of FNs in the two-hop neighborhood of FN $i$ on the physical fog network graph $G$; and the probability distribution for updating $a_i\in\mathcal{F}_i$ is:
	\begin{align}\label{samp_dist}
	p(a_i | \a_{\Gamma_i})=\frac{\exp\left(-\mathcal{L}_{\Omega_i}(\{a_i, \a_{\Gamma_i}\})/\sigma\right)}{\sum_{a^\prime\in\mathcal{F}_i}\exp(-\mathcal{L}_{\Omega_i}(\{a^\prime, \a_{\Gamma_i}\})/\sigma)}
	\end{align}
	where $\mathcal{L}_{\Omega_i}= \sum\limits_{m\in\bigcup_{k\in\Omega_i}\mathcal{M}_k} V D_m(\a,\tilde{\b})+ \sum\limits_{n\in\Omega_i} q_nE_n(\a,\tilde{\b})$, and $\Omega_i$ is the one-hop neighborhood of FN $i$, including FN $i$ itself.
\end{proposition}
\begin{proof}
See online Appendix B \cite{onlineappendix}.
\end{proof}

Proposition 1 implies that when FN $i$ changes its decision $a_i$, the change in $\mathcal{L}_{\Omega_i}$ is the same no matter how a FN outside of $\Gamma_i$ also changes its configuration at the same time. This property is the key to enabling the parallelism in GS. Moreover, we know from \eqref{samp_dist} that a decision is selected with a higher probability if it leads to a lower cost of FNs in $\Omega_i$.

Given this, we divide all FNs into $L \leq N$ groups such that no two FNs within a group are in each other's Markov Blanket and hence, FNs within the same group can update their configuration in parallel. Intuitively, we would like to minimize $L$ in order to achieve the maximal level of parallelization. Finding the minimum value of $L$ is equivalent to a graph coloring problem on the MRF. Suppose a MRF is colored with $L$-coloring, each FN will be assigned one of $L$ colors and FNs in its Markov Blanket will have a different color. The colored MRF ensures that all FNs with the same color are conditionally independent of each other in configuration update. Let $l_i$ denote the FNs in color $i\in L$.

\subsection{Distributed Algorithm based on CPGS}
\begin{figure*}[b]
	\centering
	\includegraphics[width=1 \linewidth]{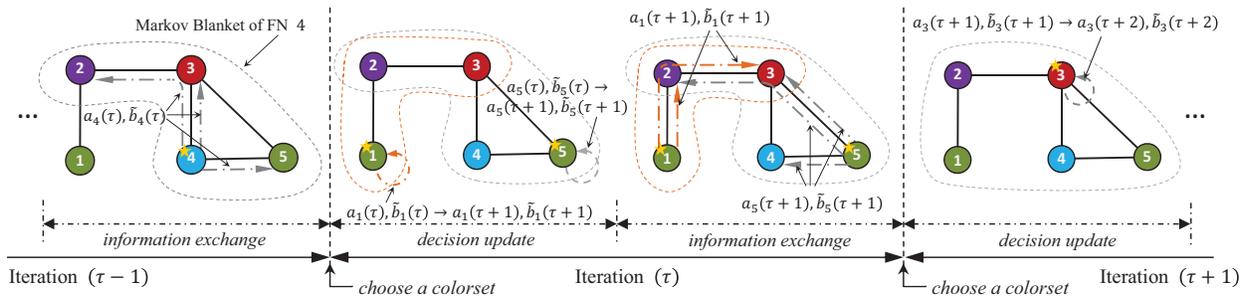}
	\caption{Illustration of CPGS for Fog configuration. It depicts a physical network with 5 FNs and the \textit{star sign} denotes the chosen colorset in each iteration.}
	\label{alg_illu}
\end{figure*}
Now, we present the distributed algorithm to solve \textbf{P2} based on CPGS (Algorithm \ref{dis_alg_p2}). The algorithm works in an iterative manner as illustrated in Fig. \ref{alg_illu}. In each iteration $\tau$, a colorset $l_j$ is chosen according to a prescribed order or randomly. Each FN $i \in l_j$ goes through two steps: \textit{decision update} and \textit{information exchange}. To update the hosting decision, FN $i$ needs two pieces of information: (1) the service demand of its one-hop neighbors, which is exchanged at the beginning of each time slot $t$; (2) the current hosting decision of FNs within its \textit{Markov Blanket} $\Gamma_i$, which are exchanged in the previous iteration. With this information, FN $i$ computes $\mathcal{L}_{\Omega_i}(a_i,\a_{\Gamma_i}(\tau)), \forall a_i\in\mathcal{F}_i$ locally by fixing the decisions of FNs in $\Gamma_i$. Then, FN $i$ samples a new $a_i(\tau+1)$ according to the probability distribution in \eqref{samp_dist}. After the hosting decisions are updated, the chosen FNs send new decisions to the FNs in $\Gamma_i$, which prepares for the next iteration. Note that during the iterations, FNs do not need to actually change configuration, which is only needed after the completion of CPGS.
\begin{algorithm}[t]
	\caption{Chromatic Parallel Gibbs Sampling}
	\begin{algorithmic}[1]
		\State \textbf{Input}: Service demand $\d$, $L$-colored MRF.
		\For {each iteration $\tau$}
		\parState {Choose a colorset of FNs $l_j$ randomly;}
		\For {each FN $i \in l_j$}
		\For {each feasible hosting decision $a_i \in \mathcal{F}_i$}
		\parState{Determine the association strategies of SNs within the coverage of $\Omega_i$;}
		\parState{Compute $\mathcal{L}_{\Omega_i}(\{a_i, \a_{\Gamma_i}(\tau)\})$ by fixing the configuration of FNs in $\Gamma_i,\forall i\in l_\tau$;}
		\EndFor
		\parState {Sample $a_i(\tau+1)$ according to \eqref{samp_dist} and send $a_i(\tau+1)$ to FNs in $\Gamma_i$;}
		\EndFor
		\State Stop iteration if service delay cost converges; 
		\EndFor
		\State \textbf{Return}: optimal configuration profile $\a^*$.
	\end{algorithmic}\label{dis_alg_p2}
\end{algorithm}

\subsection{Performance Analysis of CPGS}
Next, we prove the convergence and optimality of CPGS.

\begin{proposition} [Convergence and Optimality] \label {theo:converge}
	CPGS converges from any initial distribution to Gibbs distribution $\pi(\a)$
	\begin{align}
	\pi(\a)=\dfrac{e^{-\mathcal{L}(\a)/\sigma}}{\sum_{\a\in\mathcal{F}_1 \times \dots \times \mathcal{F}_N} e^{-\mathcal{L}(\a)/\sigma}},
	\end{align}
	and as $\sigma \to 0$, CPGS converges to the global optimum with probability 1.	
\end{proposition}
\begin{proof}
	See online Appendix C \cite{onlineappendix}.
\end{proof}

Following the classic result of parallel computing in \cite{bertsekas1989parallel}, we can analyze the time complexity of CPGS:
\begin{proposition} [Time complexity]\label{theo:runtime}
	Given a $L$-coloring of the MRF, CPGS generates a new configuration profile for the FN network in a runtime complexity of $O(L)$.
\end{proposition}
\begin{proof}
     See online Appendix D \cite{onlineappendix}.
\end{proof}

Proposition \ref{theo:runtime} indicates that CPGS advances the sampling chain for an $L$-coloring FN network in runtime $O(L)$ rather than $O(N)$. Typically, $L$ is much smaller than $N$, thereby accelerating the convergence speed of AFC.

\section{Simulation}\label{sec:simulation}
In this section, we evaluate the performance of AFC with Matlab simulations. We consider a 500m$\times$500m industrial plant served by 16 battery-powered FNs deployed with mesh layout. The service hosting capacity of FNs is set as $C=2$. Each FN has a long-term energy constraint $Q_n=10$ W$\cdot$h and unit energy consumption for each FN is set as $\kappa_n=6\times10^{-9}$ W$\cdot$h. The serving radius of FNs is set as 120m which creates co-coverage in our problem. The SNs are randomly scattered in the network using a homogeneously Poisson Process with density $0.1$, which generates a total of 52 SNs. Each SN is randomly assigned with one service type $\theta(m)$ from a total of $K = 6$ services. For each type-$k$ service, the input data size and required CPU cycles for one task are randomly drawn from $\gamma \in [0.5,1]$MB and $\mu_k\in[50,200]$M, respectively. The decision cycle for fog configuration (length of time slot) is set as 1 min. At the beginning of each time slot, FNs decide configurations with proposed algorithm and broadcast to SNs. Each SN observe the service availability at reachable FNs and determine SN-FN based on the fog configurations and wireless channel condition. The wireless channel condition is determined by the free space path-loss $L[\text{dB}] = 42.6 + 26\log_{10}(d[\text{km}])+20 \log_{10}(2,400[\text{MHz}])$ and the shadowing which is randomly drawn from a normal distribution $N(0,16)$. Till the end of the time slot, SNs can send tasks to the associated FN. The task generation of SNs follows a Poisson process with arrival rate $d^t_m\in[0,10]$. The FNs process the received tasks locally or offload to the cloud server via backbone Internet depending on the service availability and task admission decisions. The CPU frequency is 2 GHz at FNs and 4 GHz at cloud. The backbone transmission rate is $r^t \in [2,6]$Mb/s and the round-trip time is $h=200$ms. The proposed algorithm is compared with three benchmarks:
\begin{itemize}
	\item \emph{Delay-optimal Fog configuration} (D-optimal): FNs are configured to minimize the service delay regardless of the long-term energy constraints. It is a combinatorial optimization and can be solved by Gibbs sampling with simulated annealing \cite{welsh1989simulated}.
	\item \emph{Non-cooperative Fog configuration} (NCOP): Each FN works independently to serves a dedicated set of SNs. Therefore, FNs simply host the services with the largest expected demand. The long-term energy consumption constraints are enforced by Lyapunov optimization \cite{neely2010stochastic}.
	\item Single-slot constraint (SSC): Instead of following the long-term constraints, FNs impose an energy constraint $E^t_n(\a^t,\b^t)\leq Q_n$ in each time slot $t$ such that the long-term energy constraints are satisfied. SSC is fully decoupled temporally and therefore can be solved by Constraint Gibbs sampling \cite{ermon2012uniform}.
\end{itemize}
\begin{figure}[htb]
	\centering	
	\vspace{-0.15 in}
	\subfigure[Time average delay]{\label{fig:Tave_delay}
		\includegraphics[width=0.55\linewidth]{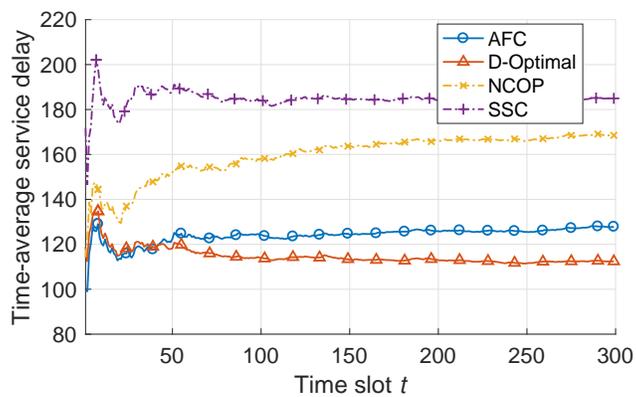}}
	\subfigure[Time average energy consumption]{\label{fig:Tave_energy}
		\includegraphics[width=0.55\linewidth]{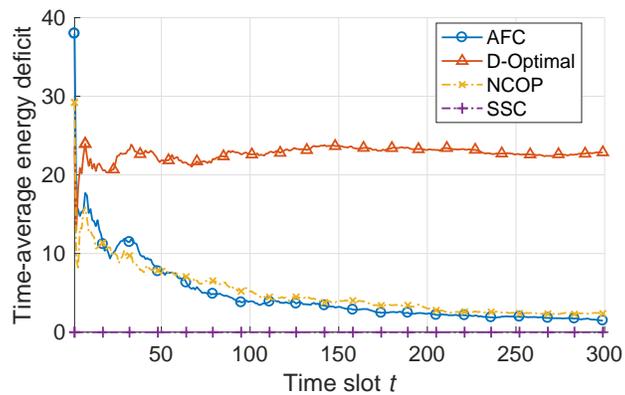}}
	\caption{Run-time performance comparison}
	\label{fig:performance_comp}
\end{figure}

\subsection{Run-time Performance Comparison}
Fig.\ref{fig:Tave_delay} and Fig. \ref{fig:Tave_energy} depict the time-average service delay and time-average energy deficit, respectively.
It can be seen that AFC achieves the close-to-optimal service delay while closely following the long-term energy constraints. Specifically, D-optimal achieves the lowest service delay since it is designed to minimize the service delay by fully exploiting the computation resource at FNs regardless of the energy constraints. As a result, D-optimal incurs a large amount of energy deficit as shown in Fig. \ref{fig:Tave_energy}. The main purpose of AFC is to follow the long-term energy constraint of each FN while minimizing the service delay. As can be observed in Fig. \ref{fig:Tave_energy}, the time-average energy deficits of AFC and NCOP converge to zero, which means that the long-term energy constraints are satisfied. Moreover, AFC achieves a close-to-optimal delay performance. By contrast, NCOP incurs a large service delay with cooperation removed. The SSC scheme poses an energy constraint in each time slot thereby satisfying the long-term energy constraints. However, SSC makes the energy scheduling less flexible across time slots and hence cannot handle well the temporal variation of service demand and results in a large service delay.

\subsection{Service Demand Allocation}
Next, we proceed to see how AFC works to benefit the Fog system. The key idea of AFC is to accommodate more service demand at FNs thereby avoiding the large service delay incurred by Cloud offloading. Fig. \ref{fig:demanddist} depicts the allocation of service demand in the first 20 time slots when running AFC and NCOP. We see clearly that AFC allows more service demand to be processed within the Fog system by enabling the collaborative service hosting. By contrast, NCOP relies more on the cloud server to process SNs' service demand.
\begin{figure}[htb]
	\centering
	\includegraphics[width=0.55\linewidth]{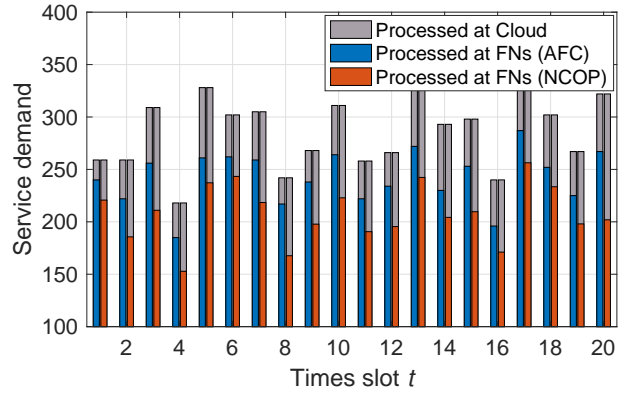}
	\caption{Service demand allocation}
	\label{fig:demanddist}
\end{figure}

\subsection{Impact of control parameter $V$}
Figure \ref{fig:tradeoff} shows the impact of control parameter $V$ on the performance of AFC. The result presents a $[O(1/V), O(V)]$ trade-off between the long-term system delay cost and the long-term energy deficit, which is consistent with our theoretical analysis in Theorem \ref{theo:AFC_bound}. With a larger $V$, AFC emphasizes more on the system delay cost and is less concerned with the energy deficit. As $V$ grows to the infinity, AFC is able to achieve the optimal delay cost.
\begin{figure}[htb]
	\centering
	\includegraphics[width=0.55\linewidth]{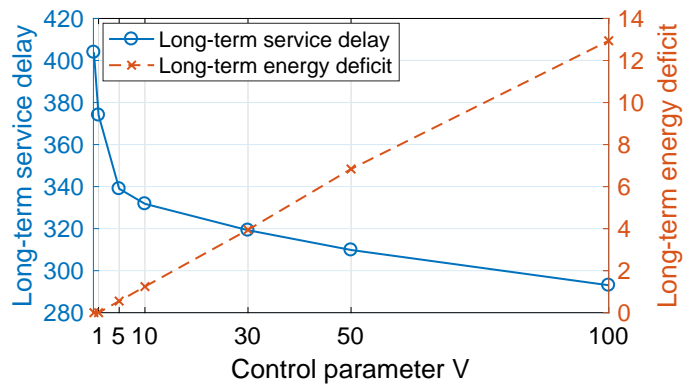}
	\caption{Impact of control parameter $V$}
	\label{fig:tradeoff}
\end{figure}

\subsection{Service Hosting Capacity}
Fig. \ref{fig:varycapacity} depicts the long-term service delay achieved by AFC and NCOP with different service hosting capacities. It can be observed that the service delay decreases with the increase in the service hosting capacity for both AFC and NCOP. This is due to the fact that more service demand can be satisfied by the FNs without sending to the remote cloud. Moreover, by comparing AFC and NCOP, we see that the delay reduction achieved by AFC decreases as the service hosting capacity grows. This is because most services can be satisfied by an individual FN with large service hosting capacity, and therefore the role of collaboration is diminished.

\begin{figure}[htb]
	\centering
	\includegraphics[width=0.55\linewidth]{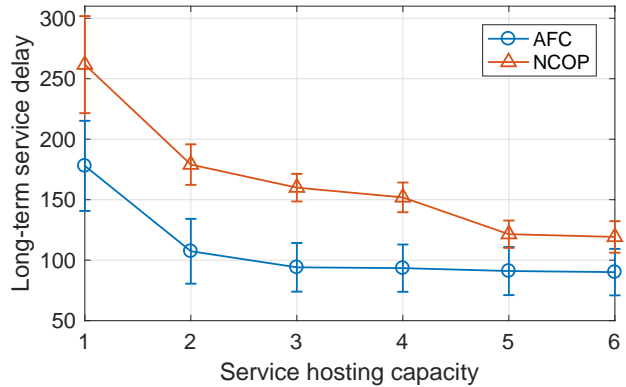}
	\caption{Impact of service hosting capacity}
	\label{fig:varycapacity}
\end{figure}

\subsection{Convergence of CPGS}
Fig. \ref{fig:cpgs} compares the convergence processes of CPGS and sequential GS for one time slot. It can be observed that CPGS converges much faster than the sequential GS. In this particular example, CPGS converges in 10 iterations, by contrast, the sequential GS takes 20 iterations to converge. Moreover, we see that CPGS and the sequential GS converges to the same optimal value, which means that CPGS preserves the features of ergodicity and optimality of sequential GS.
\begin{figure}[htb]
	\centering
	\includegraphics[width=0.55\linewidth]{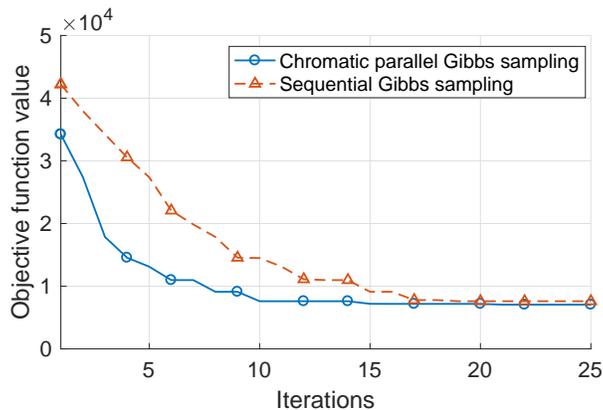}
	\caption{Convergence of CPGS}
	\label{fig:cpgs}
\end{figure}

\section{Conclusion} \label{sec:conclusion}
In this paper, we studied adaptive Fog configuration under energy constraints for IIoT systems. We proposed AFC, an online distributed algorithm for Fog configuration adaptive to both temporal and spatial service demand patterns. The proposed algorithm is easy to implement and provides provable performance guarantee. This work makes a valuable step towards optimizing the performance of Fog systems by considering service availability at fog nodes. However, future efforts are required to put the proposed framework into real-world application. For example, real demand traces of industrial things are preferred for the algorithm evaluation; practical issues such as communication failure should be considered during the distributed optimization; real IIoT platform can be constructed to run and verify the efficacy of the proposed framework.

\bibliographystyle{IEEEtran}
\bibliography{refs}

\end{document}